\def\final{1}
\documentclass[runningheads, letterpaper, 11pt]{article}
\usepackage{graphicx}
\usepackage{algorithmic}
\usepackage{algorithm}
\usepackage{array,float}
\usepackage[usenames,dvipsnames]{color}
\usepackage{amstext,amssymb,amsmath}
\usepackage{amsthm, verbatim,bm}
\usepackage{geometry}
\usepackage{times}
\usepackage{hyperref}
\usepackage{pdfsync}

\newenvironment{proofof}[1]{\begin{proof}[Proof of {#1}]
}{\end{proof}}

\newtheorem{theorem}{Theorem}[section]

\newtheorem{lemma}[theorem]{Lemma}

\newtheorem{claim}[theorem]{Claim}
\newtheorem{corollary}[theorem]{Corollary}
\newtheorem{definition}[theorem]{Definition}
\newcommand{\set}[1]{{\{#1\}}}
\newcommand{\D}{\mathcal{D}}
\newcommand{\T}{\mathcal{T}}
\newcommand{\A}{\mathcal{A}}
\newcommand{\Dist}{\bm{Distr}}
\newcommand{\re}{\mathbb{R}}
\newcommand{\liptest}{\bm{Liptest}}
\newcommand{\floor}[1]{\left\lfloor{#1} \right\rfloor}

\newcommand{\m}{\mathcal}
\newcommand{\atest}{\mathcal{A}_{\bm{test}}}

\renewcommand{\comment}[1]{}

\ifnum\final=0
\newcommand{\mjnote}[1]{\footnote{{\color{blue}{\bf MJ:} #1}}}
\newcommand{\kdnote}[1]{\footnote{{\color{cyan}{\bf KD:} #1}}}
\newcommand{\atnote}[1]{\footnote{{\color{red}{\bf AT:} #1}}}
\else
\newcommand{\mjnote}[1]{}
\newcommand{\kdnote}[1]{}
\newcommand{\atnote}[1]{}
\fi

\begin{document}
	\title{\textbf{Testing Lipschitz Property over Product Distribution and its Applications to Statistical Data Privacy}}
	
\author{
Kashyap Dixit\\
Pennsylvania State University\\
kashyap@cse.psu.edu
\and
Madhav Jha\\
Pennsylvania State University\\
mxj201@cse.psu.edu
\and
Abhradeep Thakurta\\
Pennsylvania State University\\
azg161@cse.psu.edu
}
	
    \date{}
	\maketitle

\begin{abstract}
Analysis of statistical data privacy has emerged as an important area of research. In this work we design algorithms to test  privacy guarantees of a given Algorithm $\A$ executing on a data set $\D$ which contains potentially sensitive information about individuals. We design an efficient algorithm $\atest$ which can verify whether $\A$ satisfies \emph{generalized differential privacy} guarantee. Generalized differential privacy \cite{BBGLT11} is a relaxation of the notion of differential privacy initially proposed by \cite{DMNS06}. By now differential privacy is the most widely accepted notion of statistical data privacy.

To design Algorithm $\atest$, we show a new connection between the differential privacy guarantee and Lipschitzness property of a given function. More specifically, we show that an efficient algorithm for testing of Lipschitz property can be transformed into $\atest$ which can test for generalized differential privacy. Lipschitz property testing and its variants, first studied by \cite{JR11}, has been explored by many works \cite{JR11, AJMR12, AJMR-filters12, CS12} because of its intrinsic connection to data privacy as highlighted by \cite{JR11}. To develop a Lipschitz property tester with an explicit application in privacy has been an intriguing problem since the work of \cite{JR11}. In our work, we present such a direct application of lipschitz tester to testing privacy . We provide concrete instantiations of Lipschitz testers (over both the hypercube and the hypergrid domains) which are used in $\atest$ to test for privacy of Algorithm $\A$ when the underlying data set $\D$ is drawn from the hypercube and the hypergrid domains respectively.

Apart from showing a direct connection between testing of privacy and Lipschitzness testing, we generalize the work of \cite{JR11}  to the setting of distribution property testing. We design an efficient Lipschitz testing algorithm when the distribution over the domain points is not uniform. More precisely, we design an efficient Lipschitz tester for the case where the domain points are drawn from hypercube according to some fixed product distribution. This result is of independent interest to the property testing community. It is important to note that to the best of our knowledge our results on Lipschitz testing over product distributions is the only positive result in property testing literature for non-uniform distributions after \cite{AilonC06}.
\end{abstract}

\setcounter{page}{0}

\newpage

\section{Introduction}
\label{sec:intro}
Consider a data sharing platform like \emph{BlueKai}, \emph{TellApart} or \emph{Criteo}. These platforms extensively collect and share user data with third-parties (e.g., advertisers) to enhance specific user experience (e.g., better behavioral targeting). Now, the third party applications use these data to train their machine learning algorithms for better prediction abilities. Since, the data which gets shared is extremely rich in user information, it immediately poses privacy concerns over the user information \cite{Kor2010,CKNFS11}. One way to address the privacy concerns due to the third-party learning algorithms is to train the third party algorithms ``in-house'', i.e., within the data sharing platform itself thus, making sure that the trained machine learning model preserves privacy of the underlying training data. In this paper, we study a theoretical abstraction of the above mentioned problem.

Let $\D$ be a data set where each record corresponds to a particular user and contains potentially sensitive information about the user (for example, the click history of the user for a set of advertisements displayed). Let $\A$ be an algorithm that we would like to execute on the data set $\D$ (possibly to obtain some global trends about the users in $\D$) without compromising individual's privacy. This challenging problem has recently received a lot of attention in the form of theoretical investigation in determining the privacy-utility trade-offs for various old and new algorithms. However, even if an algorithm is provably ``safe'', in practice the algorithm will be implemented in a programming language that may originate from untrusted third party. This brings its own set of challenges and has primarily been addressed in the following way: transform the algorithm $\A$ into a variant which provably satisfies some theoretically sound notion of data privacy (e.g., \emph{differential privacy} \cite{DMNS06}) either by syntactic manipulation (e.g. \cite{McSherry09,ReedP10}) or doing so in some algorithmic/systems framework (eg. ~\cite{NRS, JR11, MTSSC12, RSKSW10}). While each approach has its own appeal, they all have a few shortcomings. For example, they suffer from weak utility guarantees~\cite{NRS, MTSSC12, RSKSW10} or take prohibitively large running time~\cite{JR11}  or require use of specialized syntax \cite{McSherry09,ReedP10} making it somewhat nontrivial for a non-privacy expert to produce an effective transformation.

In this work, we take a new approach to the above problem which we call {\em privacy testing}. Specifically, we initiate the study of {\em testing} whether an input algorithm $\A$ satisfies statistical privacy guarantees. We do this by formulating the problem in the well-studied framework of property testing \cite{RS, GGR}.



\paragraph{Privacy testing}
Before we execute an Algorithm $\A$ which claims to satisfy a pre-approved notion of privacy, we test for the validity of such a claim. To the best of our knowledge, ours is the first work to study this approach. More precisely, in this work we initiate the study of testing an algorithm $\A$ for differential privacy guarantees. Differential privacy in the recent past has become a well established notion of privacy \cite{Dwork06,Dwork08,Dwork09}. Roughly speaking, differential privacy guarantees that the output of an algorithm $\A$ will not depend ``too much'' on any particular record of the underlying data set $\D$. We design testing algorithms to test whether $\A$ satisfies \emph{generalized} differential privacy \cite{BBGLT11} or not. \emph{Generalized differential privacy} is a relaxation of differential privacy and follows the same principles as differential privacy. Under specific setting of parameters, generalized differential privacy collapses to the definition of differential privacy.  For a precise definition, see Section \ref{sec:gendp}. It seems to us (and we make it more formal later on) that it may not be possible to design a computationally efficient testing algorithm for testing the notion of exact differential privacy, since in some sense it is a worst case notion privacy (see \cite{BBGLT11,BD12} for a discussion on this).

\paragraph{Testing Lipschitz property under product distribution and its connection to privacy testing}
The goal of testing properties of functions is to distinguish between functions which satisfy a given property from functions which are ``far'' from satisfying the property. The notion of ``far'' is usually the fraction of points in the domain of the function on which the function needs to be redefined to make it satisfy the property.

To test for generalized differential privacy, we show a new connection between differential privacy and the problem of testing Lipschitz property which was first studied by \cite{JR11}. A recent line of work \cite{JR11, AJMR12, AJMR-filters12} has sought to explore applications of sublinear algorithms (specifically, \emph{property testers} and \emph{reconstructors}) to data privacy. We continue this line of work and show the first application of property testers (which are vastly more efficient than property reconstructors) to the setting of data privacy. Indeed, prior to this work it was not clear if property testers for Lipschitz property can be used at all in data privacy setting.

Let $\T$ be the universe from which data sets are drawn where each data set has the same number of records. A function $f: \T \rightarrow \mathbb{R}$ is $\alpha$-Lipschitz if for all pair of points $x, x^{\prime} \in \T$ the following condition holds: $|f(x) - f(x^{\prime})| \leq d_H(x, x^{\prime})$ ~where $d_H$ is the Hamming distance between $x$ and $x^{\prime}$ (that is, $d_H(x,x^{\prime})$ is the number of entries in which $x$ and $x^{\prime}$ differ). To define Lipschitz tester, we define the notion of distance between functions $f$ and $g$ defined on the same (finite) domain $\T$ under distribution $\Dist$ as follows: $dist(f,g) \stackrel{\mbox{def}}{=} \Pr\limits_{x \sim \Dist}[f(x) \neq g(x)]$. A Lipschitz tester gets an oracle access to function $f$, a distance parameter $\epsilon \in (0,1]$. It accepts Lipschitz functions $f$ and rejects with high probability functions $f$ which are $\epsilon$-far from Lipschitz property. Namely, functions $f$ for which $\min dist(f,g) > \epsilon$, where the minimum is taken over all Lipschitz functions $g$. In this work, we extend the result of \cite{JR11} to the setting of product distribution.

While $\Dist$ is usually taken to be the uniform distribution in the property testing literature, in our setting it will be important to allow $\Dist$ to be more general distribution. Taking $\Dist$ to be something other than uniform distribution is challenging to investigate even for the special case of product distributions. Indeed, prior to this work the only positive result known for the product distribution setting is the work by \cite{AilonC06} for monotonicity testing. For the setting where $\Dist$ is an arbitrary unknown distribution there are exponential lower bounds on computational efficiency of the tester are known \cite{HK07}. Above result is stated for functions with discrete range of the form $\delta \mathbb{Z}$.

 In this paper, we show that one can use a Lipschitz property testing algorithm ($\liptest$) as a proxy for testing generalized differential privacy. The tester $\liptest$ should be able to sample \emph{efficiently} the data set according to a given probability distribution defined over domain of these data sets (see Definition \ref{def:gdp}). It has been shown that this additional requirement is sufficient to give strong privacy guarantees for the algorithm being tested.( For further details see Section \ref{sec:privacy}.) Additionally, for practical applications, this tester should run efficiently, especially over the large data set domain.

With the above  motivation in mind, we have designed such a Lipschitz tester with sub-linear time complexity (with respect to the domain size) for the hypercube domain $\T=\{0,1\}^d$ with product distribution defined on data sets in $\T$. (For further details, we refer the reader to Section \ref{sec:hypercube}.) With this construction, we can test the privacy guarantees of an algorithm in time that is poly-logarithmic in domain size.
\atnote{Mention about the hypergrid result}

\subsection{Related Work}
\label{sec:related}

In the last few years, various notions of data privacy have been proposed. Some of the most prominent are $k$-anonymity \cite{Sweeney}, $\ell$-diversity \cite{MK}, differential privacy \cite{DMNS06}, noiseless privacy \cite{BBGLT11}, natural differential privacy \cite{BD12} and generalized differential privacy  \cite{BBGLT11}. While ad-hoc notions like $k$-anonymity and $\ell$-diversity being broken \cite{GKS08}, privacy community has pretty much converged to theoretical sound notions of privacy like differential privacy.  In this paper, we work with the definition of generalized differential privacy (GDP), which is a generalization of differential privacy, noiseless privacy and natural differential privacy. The primary difference between GDP and the other related definitions is that it incorporates both the randomness in the underlying data set $\D$ and the randomness of the Algorithm $\A$, where as other notions consider either the randomness of the data or the randomness of the algorithm.

In this paper, we design algorithms ($\atest$) to test whether a given algorithm $\A$ satisfies GDP. In all our algorithms, we assume that $\A$ is given as a ``white-box'', i.e., complete access to the source code of $\A$ is provided. In this paper, all the instantiations of $\atest$ are probabilistic and use Lipschitz property testing algorithms as underlying tool set. On a related note, in the field of formal verifications there have been recent works \cite{ReedP10} using which one can guarantee that a given algorithm $\A$ satisfy differential privacy. The caveat of these kind of static analysis based algorithms is that it needs the source code for $\A$ to be written in a type-safe language which is hard for a non-expert to adapt to.

One of the primary reason for considering the sublinear (with respect to the domain size) time Lipschitz testers is the large size of domain often encountered in the study of statistical privacy of databases. The property testers (\cite{RS96,GGR98}) have been extensively studied for various approximation and decision problems. They are of particular interest because they usually have sublinear (in input size) running time which is of particular interests in the problem with large inputs. Some of the ideas and definitions in this paper have been taken from the work on distribution testing (\cite{HK07, GS09,AilonC06}). Lipschitz property testers were introduced in \cite{JR11} (which gave the explicit tester for the hypercube domain) and have since then been studied in \cite{AJMR12, AJMR-filters12} for the hypergrid domain. Recently \cite{CS12} have proposed an optimal Lipshcitz tester for the hypercube domain with the underlying distribution being uniform. 
\subsection{Our Contributions}
\label{sec:contributions}

\begin{itemize}
    \item \textbf{Formulate testing of data privacy property as Lipschitz property testing:} In this paper we initiate the study of testing privacy properties of a given candidate algorithm $\A$. The specific privacy property that we test is \emph{generalized differential privacy} (GDP) (see Definition \ref{def:gdp}). In order to design a tester for GDP property, we cast the problem of testing GDP property as a problem of testing Lipschitzness. (See Theorem~\ref{thm:utilpriv}.) The problem of testing Lipschitzness was initially proposed by \cite{JR11}.
    \item \textbf{Design a generic transformation to convert an Algorithm $\A$ to its GDP variant:} We design a generic transformation to convert a candidate algorithm $\A$ to its generalized differentially private variant. (See Theorem~\ref{thm:priv1}.)
    \item \textbf{New results for Lipschitz property testing:} In order to allow our privacy tester to be effective for a large class of data generating distributions, we extend the existing results of Lipschitz property testing to work with product distributions. We give the first efficient tester for the Lipschitz property for the hypercube domain which works for arbitrary product distribution. (See Theorem~\ref{thm1}.) Previous works (even for other function properties) have mostly focused on the case of uniform distribution. To the best of our knowledge this is the only non-trivial positive result in property testing over arbitrary product distribution apart from the result of \cite{AilonC06} on monotonicity testing.

   \item \textbf{Concrete instantiation of privacy testers based on old and new Lipschitz testers} We instantiate privacy tester using Lipschitz tester described in the previous item to get a concrete instantiation of privacy tester. This also leads to a concrete instantiation of Item 2 mentioned above. We also instantiate privacy testers based on known Lipschitz testers in the literature. This is summarized in Section~\ref{sec:hypergrid}.

\end{itemize}



\subsection{Organization of the paper}
In Section \ref{sec:prelim}, we introduce the notions of privacy used in this paper, namely, differential privacy and generalized differential privacy. We also introduce the concepts of general property testing and the specific instantiation of Lipschitz property testing. In Section \ref{sec:privacy}, we show the formal connection between testing of generalized differential privacy (GDP) and Lipschitz property testing. In Section \ref{sec:hypercube}, we state our new results of Lipschitz property testing over product distributions in the hypercube domain. In Section \ref{sec:hypergrid}, we show that Lipschitz testers over the hypergrid domain can be used to test for GDP when the data sets are drawn uniformly from the hypergrid domain. Lastly, in Section \ref{sec:disc} we conclude with discussions and open problems. 
\section{Preliminaries}
\label{sec:prelim}

\subsection{Differential Privacy and Generalized Differential Privacy}
\label{sec:gendp}

In the last few years, differential privacy \cite{DMNS06} has become a well-accepted notion of statistical data privacy in the data privacy community. At a high-level the definition of differential privacy implies that the output of a differentially private algorithm will be ``almost'' the same from an adversary's perspective irrespective of an individual's presence or absence in the underlying data set. The reason that it is a meaningful notion is because the presence or absence of an individual in the data set does not affect the output of the algorithm ``too much''. This high-level intuition can be formalized as below:

\begin{definition}[$(\alpha,\gamma)$-Differential Privacy \cite{DMNS06,ODO}]
A randomized algorithm $\mathcal{A}$ is $(\alpha, \gamma)$-differentially private if for any two data sets $\D$ and $\D'$ drawn from a domain $\T$ with $|\D\Delta \D'|=1$ ($\Delta$ being the symmetric difference), and for all measurable sets $\mathcal{O}\subseteq Range(\mathcal{A})$ the following holds:
$$\Pr[\mathcal{A}(\D)\in \mathcal{O}]\leq e^\alpha\Pr[\mathcal{A}(\D')\in \mathcal{O}]+\gamma$$.
\label{def:DP}
\end{definition}

In the above definition if $\gamma=0$, we simply call it $\alpha$-differential privacy.
In this paper we intend to test if an algorithm $\A$ is $\alpha$-differentially private. In order to test the above, we mould the problem into a problem of testing Lipschitzness over the probability measure induced by  Algorithm $\A$ over a finite set $S$ (see Section \ref{sec:privacy} for more discussion on this). Since, we want to test Lipschitzness efficiently with respect to the size of the set $S$, we will use a relaxed notion of differential privacy called \emph{generalized differential privacy} (GDP) \cite{BBGLT11}. The main idea behind GDP is that it allows us to incorporate the randomness over the data generating distribution. This in turn allows us to incorporate the failure probability of the Lipschitzness testing algorithm (over the randomness of the data generating distribution). The definition of GDP below is a slight modification to the definition proposed in \cite{BBGLT11} and in most natural settings is stronger than \cite{BBGLT11}.

\begin{definition}[$(\alpha,\gamma,\beta)$-Generalized Differential Privacy]
Let $\bm{Dist}$ be the distribution over the space of all data sets drawn from domain $\T$. Let $W\subseteq\T$ be a set such that $\Pr_{\D\sim\Dist}[\D\in W]\leq \beta$. A randomized algorithm $\mathcal{A}$ is $(\alpha,\gamma, \beta)$-generalized differentially private (GDP) if for any pair data sets $\D,\D'\in\T\setminus W$ with $|\D\Delta \D'|=1$ ($\Delta$ being the symmetric difference) and for all measurable sets $\mathcal{O}\subseteq Range(\mathcal{A})$ the following holds:
$\Pr[\mathcal{A}(\D)\in\mathcal{O}]\leq e^\alpha\Pr[\mathcal{A}(\D')\in \mathcal{O}]+\gamma$, where the probability is over the randomness of the Algorithm $\A$.
\label{def:gdp}
\end{definition}
It is worth mentioning here that the above definition generalizes the \emph{noiseless privacy} definition \cite{BBGLT11} and \emph{natural differential privacy} definition \cite{BD12} in the literature. While in both noiseless and natural differential privacy definitions the randomness is \emph{solely} over the data generating distribution $\bm{Dist}$, in GDP the randomness is both over the data generating distribution and the randomness of the algorithm.

At a high-level what GDP says is that there exists a set $W$ of ``bad'' data sets where $(\alpha,\gamma)$-differential privacy condition does not hold. But the probability of drawing a data set $\D$ (over the data generating distribution $\Dist$) from $W$ is at most $\beta$ (which is usually negligible in the problem parameters). In fact if we set $\beta=0$, then we recover $(\alpha,\gamma)$-differential privacy definition (see Definition \ref{def:DP}) exactly. Similarly, it can be shown that under different choices of $(\alpha,\gamma,\beta)$ GDP implies both noiseless privacy and natural differential privacy.

\subsection{Lipschitz Property Testing}
\label{subsec:lpt}

In this work we show that efficiently testing whether an algorithm is $(\alpha,\beta,\gamma)$-generalized differentially private reduces to the problem of testing (with high success probability over the probability measure induced by  Algorithm $\A$) if the output is Lipschitz. (For further details see section, see section \ref{sec:privacy}.)

\begin{definition}
\label{def:lip}
Given a function $f:\T\rightarrow \re$ from a metric space $(\T,d_{\T})$ to $(\re,d_{\re})$, where $d_D$ and $d_R$ denote the distance function on the domain $D$ and the range $R$ respectively. The function $f$ is \textit{c-Lipschitz} if $d_{\re}(f(x),f(y))\le c\cdot d_{\T}(x,y)$.
\end{definition}

Property testing (\cite{GGR98},\cite{RS96}) is a well studied area pertaining to randomized approximation algorithms for decision problems usually having sublinear time and query complexity. At one end of the spectrum, most of the work previously done in this area assume a uniform distribution over domain elements. The other end is to consider the setting where the distribution over the domain points is not known (\cite{HK07}).

Here, we assume that the probability measure over domain elements is known and is not necessarily uniform. Although seemingly important, to the best of our knowledge, this is the first time that such a setting is explored in the lipschitz property testing. To state our results, we will need the following notation.

Let $\m{P}$ (\textit{e.g.} Lipschitzness in this case) be the property that needs to be tested over the range of function $f:D\rightarrow R$.  We define the distance of the function $f$ from $\m{P}$ as follows.

\begin{definition}\label{def:pdist}
Let $\m{P}$ and $\T$ be defined as above. The $\m{P}\textrm{-distance}$ between functions $f,g\in \mathcal{F}$ is defined by $dist_{\m{P}}(f,g)\stackrel{def}{=}\Pr_{x\sim \T}\{f(x)\neq g(x)\}	$. The $\m{P}$-distance of a function $f$ from property $\mathcal{P}$ is defined as $dist_{\m{P}}(f,\mathcal{P})= min_{g\in \mathcal{P}}dist_{\m{P}}(f,g)$. We say that $f$ is $\epsilon$-far from a property $\mathcal{P}$ if $dist_{\m{P}}(f,\mathcal{P})\ge \epsilon$.
\end{definition}

We will need the notion of the image diameter of a function $f$ for explaining our results, which, roughly speaking, is the difference between maximum and minimum values taken by $f$ on domain $\T$.

\begin{definition}[Image diameter]
The image diameter of a function $f:\T\rightarrow \re$, denoted by $ImD(f)$, is the difference between the maximum and the minimum values attained by $f$, i.e., $\max_{x\in \T} f(x)-\min_{x\in \T} f(x)$.
\label{def:imd}
\end{definition}

\section{Test for Generalized Differential Privacy}
\label{sec:privacy}


In this work we initiate the study of testing whether a given algorithm $\A$ satisfies statistical data privacy guarantees. As a specific instantiation of the problem, we study the notion of generalized differential privacy (GDP) (see Definition \ref{def:gdp}). Roughly speaking, GDP guarantee ensures that the output of Algorithm $\A$ when executed on data set $\D$ does not depend ``too much'' on any one entry of $\D$. The term ``too much'' is formalized by three parameters $\alpha$, $\gamma$ and $\beta$, where the first two parameters ($\alpha$ and $\gamma$) depends on the randomness of the Algorithm $\A$ and the parameter $\beta$ depends on the randomness of the distribution $\Dist$ generating the data. We refer to the guarantee as $(\alpha, \gamma, \beta)$-Generalized Differential Privacy (or simply $(\alpha, \gamma,\beta)$-GDP).

Given an algorithm $\A$, we design a tester $\atest$ with the following property: if the tester outputs $\bm{YES}$, then Algorithm $\A$ is $(\alpha,\gamma, \beta)$-generalized differentially private where the parameters $\beta$ and $\gamma$ can be made arbitrarily small (at the cost of increased running time). If the tester outputs $\bm{NO}$, then the Algorithm $\A$ is {\em not} $\alpha$-differentially private. We state this formally below.

\begin{theorem}[$(\theta,\alpha,\gamma,\beta)$-Privacy testing]
Let $\liptest$ be a $\theta$-approximate Lipschitz tester (see Definition \ref{def:lip-tester} below), let $\Dist$ be a distribution on the domain of datasets $\T$ and let $\A$ be an algorithm which on input $\D \sim \Dist$ outputs a value $\A(\D)$ in the finite set $\Gamma$. Suppose there is an oracle ${\cal O}_{\A}$ which for every value $o \in \Gamma$ and for every $\D\in\T$ allows constant time access to the probability measure $\mu(\A(\D) = o)$ (where the measure  is over the randomness of the algorithm $\A$). Then there exists a ``testing'' algorithm $\atest$ which on input privacy parameters $\alpha, \beta \in (0,1]$, failure probability parameter $\gamma \in (0,1]$ and access to ${\cal O}_{\A}$ and $\Dist$ satisfies the following guarantee.
\begin{itemize}
\item \textbf{(soundness)} If Algorithm $\atest$ outputs $\bm{NO}$, then the candidate algorithm $\A$ is \textbf{not} $\alpha$-differentially private.
\item \textbf{(completeness)} If Algorithm $\atest$ outputs $\bm{YES}$, then with probability at least $1-\gamma$ the candidate algorithm $\A$ is $(\alpha \theta,0, \beta)$-generalized differentially private.
\end{itemize}

The algorithm $\atest$ uses $\liptest$ as a subroutine and runs in time $O(|\Gamma|\cdot(\text{Run time of }\liptest))$.
\label{thm:utilpriv}
\end{theorem}


To prove Theorem~\ref{thm:utilpriv}, we show a new connection between testing $(\alpha,0, \beta)$-GDP and the problem of testing Lipschitz property. The study of testing Lipschitz property was initiated by \cite{JR11}. We present an algorithm $\atest$ for testing $(\alpha,0,\beta)$-GDP based on a generalization of Lipschitz tester presented in \cite{JR11}. We formally define the (generalized) Lipschitz tester below where the definition differs from the standard property testing definition (example, as used in \cite{JR11}) in two aspects: (i) we require Lipschitz testers to only distinguish between Lipschitz functions from functions which are far from $\theta$-Lipschitz functions for some fixed $\theta \geq 1$ and  (ii) we measure distance between functions (in particular, how ``far'' the function is from satisfying the property) with respect to a pre-defined probability measure $\Dist$ on the domain.

\begin{definition}[$\theta$-approximate Lipschitz tester]\label{def:lip-tester}
A $\theta$-approximate Lipschitz tester $\liptest$ is a randomized algorithm that gets as input: (i) oracle access to function $f: \T \rightarrow \mathbb{R}$; (ii) oracle access to independent samples from distribution $\Dist$ on $\T$ and (iii) parameters $\epsilon, \gamma \in (0,1]$. It outputs a $\bm{YES}$/$\bm{NO}$ value and provides the following guarantee.
\begin{itemize}
\item If $\liptest$ outputs $\bm{NO}$, then with probability 1, the function $f$ is \textbf{not} Lipschitz.
\item If $\liptest$ outputs $\bm{YES}$, then with probability at least $1 - \gamma$, there exists a set $W \subseteq \T$ such that (i) the input function $f$ is $\theta$-Lipschitz on the domain $\T \setminus W$ and (ii)  $\Pr_{\D \sim \Dist}[\D \in W] \leq \epsilon$.
\end{itemize}
\end{definition}

We remark that setting $\theta = 1$ and $\Dist$ to be the uniform distribution on $\T$ recovers the standard definition of property tester (in our case, Lipschitz tester as defined in \cite{JR11}).

In Section \ref{sec:cand}, we show that one can extend the connection between GDP and Lipschitz testing to design an algorithm $\A_{\bm{privGen}}$ which converts the candidate algorithm $\A$ in to a $(\alpha,\gamma,\beta)$-generalized differentially private algorithm.

\subsection{(Generalized) Differential Privacy as Lipschitz Property over a Probability Measure}
\label{sec:lipprob}

Consider the domain of the data sets $\T$ to be a finite set and assume that (the randomized) Algorithm $\A$, whose privacy property is to be tested, maps a data set $\D\in\T$ to another finite set $\Gamma$, i.e. any output of $\A$ is always an element in $\Gamma$. Now let us look at the privacy guarantee of GDP (see Definition \ref{def:gdp}). Ignoring the parameters $\beta$ and $\gamma$, the privacy guarantee suggests that for any pair of neighboring data sets $\D, \D^{\prime} \in\T$ (drawn from the distribution $\Dist$) and any $o\in\Gamma$, the following is true:
\begin{align}
e^{-\alpha}\mu(\A(\D')=o) \leq \mu(\A(\D)=o)\leq e^\alpha\mu(\A(\D')=o)
\label{eq:1}
\end{align}

The measure $\mu$ is the probability induced by the randomness of the Algorithm $\A$. Taking logarithm of \eqref{eq:1}, we get
\begin{align}
|\log\mu(\A(\D)=o)-\log\mu(\A(\D')=o)|\leq \alpha
\label{eq:2}
\end{align}
We will use the following formulation of \eqref{eq:2}: $|\frac{1}{\alpha}\log\mu(\A(\D)=o)-\frac{1}{\alpha}\log\mu(\A(\D')=o)|\leq d_H(\D,\D')$, where $d_H$ is the Hamming metric. Now, if we view the expression $\frac{1}{\alpha}\log\mu(\A(\D)=o)$ as a function $\lambda_o:\T\to\re$ defined by setting $\lambda_o(\D) = \frac{1}{\alpha}\log\mu(\A(\D)=o)$, then we get the following condition: $|\lambda_o(\D)-\lambda_o(\D')|\leq d_H(\D,\D')$. This condition is exactly the Lipschitzness guarantee for $\lambda_o$ under the Hamming metric. Using this observation we state the following meta-algorithm $\atest$ (Algorithm \ref{alg:certify}) to test whether given Algorithm $\A$ is $(\alpha,0,\beta)$-generalized differentially private. In Algorithm \ref{alg:certify} (Algorithm $\atest$), we use a black box Lipschitz property tester $\liptest$. Later in the paper we instantiate $\liptest$ with a specific testing algorithms.

 \begin{algorithm}[htb]
  \caption{$\atest$: Generalized Differential Privacy (GDP) tester}
  \begin{algorithmic}[1]
  \REQUIRE Algorithm $\A$, data generating distribution $\Dist$, data domain $\T$, output range $\Gamma$, privacy parameters $(\alpha,\beta)$ and failure parameter $\gamma$
  \STATE $flag\leftarrow\bm{FALSE}$
  \STATE Let $\liptest$ be a $\theta$-approximate Lipschitz tester defined in Definition~\ref{def:lip-tester}.
  \FORALL{values $o\in\Gamma$}
   \STATE Define function $\lambda_o : \T \rightarrow \mathbb{R}$ by setting $\lambda_o(\D) = \frac{1}{\alpha}\log\mu(\A(\D)=o)$.
    \STATE Run $\liptest$ on $\lambda_o$ with {\em proximity} parameter $\frac{\beta}{\Gamma}$ and failure probability parameter $\frac{\gamma}{|\Gamma|}$.\label{line:13}
    \STATE If $\liptest$ outputs $\bm{NO}$, then $flag\leftarrow\bm{TRUE}$
  \ENDFOR
  \STATE If $flag=\bm{FALSE}$, then output $\bm{YES}$, otherwise output $\bm{NO}$
  \end{algorithmic}
  \label{alg:certify}
\end{algorithm}


At a high-level Algorithm $\atest$ does the following. For each possible output $o\in\Gamma$, it defines a function table $\lambda_o$ (with the domain $\T$).  It then invokes the Lipschitz testing algorithm $\liptest$ to test $\lambda_o$ for Lipschitzness property. If for every output $o\in\Gamma$, $\liptest$ outputs $\bm{YES}$, then $\atest$ outputs affirmative, and outputs negative otherwise.

\subsubsection{Proof of Theorem~\ref{thm:utilpriv}}

The claim about the running time of Algorithm $\atest$ stated in Theorem~\ref{thm:utilpriv} follows directly from the definition of Algorithm $\atest$ (Algorithm \ref{alg:certify}).  We state and prove the soundness and completeness guarantees of Theorem~\ref{thm:utilpriv} separately as Claim~\ref{claim:soundness} and Claim~\ref{claim:completeness} respectively below.

\begin{claim}[Soundness guarantee]\label{claim:soundness}
If Algorithm $\atest$ (Algorithm \ref{alg:certify}) outputs $\bm{NO}$, then the candidate algorithm $\A$ is not $\alpha$-differentially private.
\label{cl12m}
\end{claim}
\begin{proof}
If Algorithm $\atest$ outputs a $\bm{NO}$, then there exists an $o\in\Gamma$ such that $\liptest$ outputs NO on $\lambda_o$. By defintion of $\liptest$ (see Definition~\ref{def:lip-tester}), we get that $\lambda_o$ is not Lipschitz. In other words, we have, $|\lambda_o(\D)-\lambda_o(\D')|= |\frac{1}{\alpha}\log\mu(\A(\D)=o)-\frac{1}{\alpha}\log\mu(\A(\D')=o)| > 1$. Therefore, either $\mu(\A(\D)=o)>e^{\alpha}\mu(\A(\D')=o)$ or $\mu(\A(\D)=o)<e^{-\alpha}\mu(\A(\D')=o)$, as required.
\end{proof}

\begin{claim}[Completeness guarantee]\label{claim:completeness}
If Algorithm $\atest$ (Algorithm \ref{alg:certify}) outputs $\bm{YES}$, then with probability at least $1-\gamma$ (over the randomness of $\liptest$), the candidate algorithm $\A$ is $(\alpha \theta,0, \beta)$-generalized differentially private.
\label{cl14m}
\end{claim}
\begin{proof}
If Algorithm $\A$ outputs $\bm{YES}$, then by the union bound it follows that with probability at least  $1-{\gamma}$, the following condition holds for every $o \in \Gamma$: There exists a set $W_o \subseteq \T$ such that (i) $\lambda_o$ satisfies  $\theta$-Lipschitz condition for every $\D, \D^{\prime} \in \T \setminus W_o$ and (ii)  $\Pr\limits_{x \sim \Dist}[x \in W_o] < \frac{\beta}{|\Gamma|}$.

Let $W = \displaystyle\bigcup_{o \in \Gamma} W_o$. We show that with probability at least $1-\gamma$ (over the randomness of $\liptest$), the following holds: algorithm $\A$ satisfies $\alpha \theta$-differential privacy condition on the set $\T \setminus W$ and $\Pr\limits_{\D \sim \Dist}[\D \in W] \leq \beta$.

Condition (i) above implies that for every $o \in \Gamma$, $\lambda_o$ is $\theta$-Lipschitz on $\T \setminus W$. Therefore, we get the following for every neighboring pairs of data sets $\D,\D'\in\T\setminus W$.

\begin{align*}
|\lambda_o(\D) - \lambda_o(\D^{\prime})| \leq \theta\\
\Rightarrow | \frac{1}{\alpha}\log\mu(\A(\D)=o) -  \frac{1}{\alpha}\log\mu(\A(\D^{\prime})=o)| \leq \theta \\
\Rightarrow e^{-\alpha \theta}\leq\frac{\mu(\A(\D) = o)}{\mu(\A(\D') = o)}\leq e^{\alpha \theta}
\end{align*}

Also, using  Condition (ii) and the union bound over all $o \in \Gamma$, we get the following.
$$\Pr\limits_{\D \sim \Dist}[\D \in W] \leq \sum_{o \in \Gamma} \Pr\limits_{\D \sim \Dist}[\D \in W_o] \leq \beta.$$

Since Conditions (i) and (ii) both hold with probability at least $1-\gamma$ (over the randomness of $\liptest$), we get the desired claim.

\end{proof}

\subsection{Application of GDP tester to ensure privacy for the output of a given candidate algorithm}
\label{sec:cand}

In this section we will demonstrate how one can use Algorithm $\atest$ (Algorithm \ref{alg:certify}) designed in the previous section to guarantee $(\alpha,\beta,\gamma)$-generalized differential privacy to the output produced by a candidate Algorithm $\A$. The details are given in Algorithm \ref{alg:privGen}.
The theoretical guarantees for Algorithm \ref{alg:privGen} are given below.

\begin{theorem}[$(\theta,\alpha,\gamma,\beta)$-generalized differentially private mechanism]
Let $\liptest$ be a $\theta$-approximate Lipschitz tester (see Definition \ref{def:lip-tester}) used in the testing algorithm $\atest$ (Algorithm \ref{alg:certify}). Under the same assumptions of Theorem \ref{thm:utilpriv}, following are true for Algorithm $\A_{\bm{privGen}}$ (Algorithm \ref{alg:privGen}).
\begin{itemize}
\item \textbf{(privacy)} Algorithm $\A_{\bm{privGen}}$ (Algorithm \ref{alg:privGen}) is $(\alpha\theta,\beta,\gamma)$-generalized differentially private (GDP).
\item \textbf{(utility)} If the candidate Algorithm $\A$ is $\alpha$-differentially private, then Algorithm $\A_{\bm{privGen}}$ (Algorithm \ref{alg:privGen}) always produces the output $\A(\D)$.
\end{itemize}
\label{thm:priv1}
\end{theorem}

\begin{algorithm}[htb]
  \caption{$\A_{\bm{privGen}}$: Generalized differentially private mechanism}
  \begin{algorithmic}[1]
  \REQUIRE Data set $\D$, candidate algorithm $\A$, testing algorithm $\atest$, data generating distribution $\Dist$, data domain $\T$, output set $\Gamma$, privacy parameters $(\alpha,\beta,\gamma)$
  \STATE Run $\atest$ with parameters $\A,\Dist,\T,\Gamma$, privacy parameters $(\alpha,\beta)$, and failure parameter $\gamma$
  \STATE If $\atest$ outputs $\bm{YES}$, then output $\A(\D)$, output $\bm{FAILURE}$ otherwise
  \end{algorithmic}
  \label{alg:privGen}
\end{algorithm}

\subsubsection{Proof of Theorem \ref{thm:priv1}}
The proof of Theorem \ref{thm:priv1} follows from the two claims below.
\begin{claim}[Privacy]
Algorithm $\A_{\bm{privGen}}$ (Algorithm \ref{alg:privGen}) is $(\alpha\theta,\gamma,\beta)$-generalized differentially private (GDP).
\end{claim}

\begin{proof}
First note that from Claim \ref{claim:completeness}, it follows that if Algorithm $\atest$ (Algorithm \ref{alg:certify}) outputs $\bm{YES}$, then w.p. $\geq 1-\gamma$, the candidate algorithm $\A$ is $(\alpha\theta,0,\beta)$-GDP. Now to complete the proof, we provide the following argument.

\begin{itemize}
	\item \textbf{Case 1 [Algorithm \ref{alg:privGen} outputs $\A(D)$]:}
We define event $Ev$ to be the following:  For every $o\in\Gamma$ there exists a set $W_o 		 \subseteq \T$ such that (i) $\lambda_o$ satisfies  $\theta$-Lipschitz condition for every $\D, \D^{\prime} \in \T \setminus W_o$ and (ii)  $\Pr\limits_{x \sim \Dist}[x \in W_o] < \beta$. As implied by the GDP guarantee, event $Ev$ holds with probability $1-\gamma$. Hence, we have the following for all $o\in\Gamma\cup\{\bm{FAILURE}\}$
\begin{align*}
\Pr[\A_{\bm{privGen}}(\D)=o]&\leq\Pr[\A_{\bm{privGen}}(\D)=o|Ev]\Pr[Ev]+\Pr[\bar{Ev}]\\
&\leq e^{\alpha\theta}\Pr[\A_{\bm{privGen}}(\D')=o|Ev]\Pr[Ev]+\gamma\\
&\leq e^{\alpha\theta}\Pr[\A_{\bm{privGen}}(\D')=o \wedge Ev]+\gamma\\
&\leq e^{\alpha\theta}\Pr[\A_{\bm{privGen}}(\D')=o]+\gamma
\end{align*}
	\item \textbf{Case 2[Algorithm \ref{alg:privGen} outputs $\bm{FAILURE}$]:}
	In this case, the output is trivially $(\alpha,\gamma,\beta)$-generalized differentially private since the output (i.e., $\bm{FAILURE}$) is independent of the data set $\D$.
	\end{itemize}
With this the proof is complete.
\end{proof}

\begin{claim}[Utility]
If the candidate Algorithm $\A$ is $\alpha$-differentially private, then Algorithm $\A_{\bm{privGen}}$ (Algorithm \ref{alg:privGen}) always produces the output $\A(\D)$.
\end{claim}

The proof of the above claim follows from the fact that if the candidate algorithm $\A$ is $\alpha$-differentially private, then $\atest$ will always output $\bm{YES}$.

\section{Lipschitz Property Testing over Hypercube domain}
\label{sec:hypercube}
In this section, we present a $(1+\delta)$-approximate Lipschitz tester (see Definition~\ref{def:lip-tester}) for functions defined on $\T = \set{0,1}^d$ where the notion of distance is with respect to any product distribution. Specifically, the points in the data set are distributed according to the product distribution $\Pi=Ber(p_1)\times Ber(p_2)\times...,\times Ber(p_d)$ where $Ber(p)$ denotes the Bernoulli distribution with probability $p$. For any vertex $x=(x_1,x_2,...,x_d)\in \T$, $x_i=1$ with probability $p_i$ and $0$ with probability $1-p_i$. Each vertex in $x\in \T$ has an associated probability mass $p_x=p_{i_1}\cdot p_{i_2}\cdots p_{i_k}\cdot (1-p_{j_1})\cdot(1-p_{j_2})\cdots(1-p_{j_{d-k}})$ where $k$ is the hamming weight of $x$, also denoted by $H(x)$ and $i_1,i_2,...,i_k$ denote the indices of $x$ with bit-value $1$.

In this section, we prove the following theorem which gives a $1$-approximate Lipschitz tester for $\delta \mathbb{Z}$-valued functions. A function is $\delta\mathbb{Z}$ valued if it produces outputs in integral multiples of $\delta$.
\begin{theorem}
\label{thm1}
Let $\T = \{0,1\}^d$ be the domain from which the data set are drawn according to a product probability distribution $\Pi=Ber(p_1)\times Ber(p_2)\times...,\times Ber(p_d)$. The Lipschitz property of functions$f:\T\rightarrow \delta\mathbb{Z}$ on these data sets can be tested non-adaptively and with one sided error probability $\omega$ in $O(\frac{d\cdot\min\{d, ImD(f)\}}{\delta(\epsilon - d^2 \delta)}\ln(\frac{2}{\omega}))$ time for $\delta\in (0, 1]$. Here $ImD$ is the image diameter defined in Definition \ref{def:imd}.
\end{theorem}

Following is an easy corollary of the above giving a $(1+\delta)$-approximate Lipschitz tester for $\mathbb{R}$-valued functions.
\begin{corollary}[of Theorem \ref{thm1}]
\label{cor:tester-for-real-range}
Let $\T = \{0,1\}^d$ be the domain from which the data set are drawn according to a product probability distribution $\Pi=Ber(p_1)\times Ber(p_2)\times...,\times Ber(p_d)$. There is an algorithm that on input parameters $\delta \in (0, 1], \epsilon \in (0, 1), d$ and oracle access to a function $f : \{0,1\}^d \rightarrow \mathbb{R}$ has the following behavior: It accepts if $f$ is Lipschitz and rejects with probability at least $1 - \omega$ if $f$ is $\epsilon$-far (with respect to the distribution $\Pi$) from $(1 + \delta)$-Lipschitz and runs in $O(\frac{d\cdot\min\{d, ImD(f)\}}{\delta (\epsilon - d^2 \delta)}\ln(\frac{2}{\omega}))$  time. Here $ImD$ is the image diameter defined in Definition \ref{def:imd}.
\end{corollary}

The proof of above theorem and corollary appears in Section~\ref{sec:thmcor}. To state the proof we need the following technical result.


We define a distribution on edges of the hypercube where the probability mass of an edge $\set{x,y}$ is given by $\frac{p_x + p_y}{d}$. Note that $\sum_{(x,y)\in E(H_d)}\frac{(p_x+p_y)}{d}=1$. Thus the probability distribution (we call it $D_E$ henceforth) on the edges defined above is consistent. Our tester is based on detecting violated edges (that is, edges which violate Lipschitz property) sampled from distribution $D_E$. Our main technical lemma (Lemma~\ref{lem1}) gives a lower bound on the probability of sampling a violated edge according to distribution $D_E$ for a function that is $\epsilon$-far from Lipschitz. (Recall that $\epsilon$-far is measured with respect to the distribution $\Pi$.)



\begin{lemma}
\label{lem1}
Let function $f:\{0,1\}^d\rightarrow \delta\mathbb{Z}$ be $\epsilon$-far from Lipschitz. Then
\begin{eqnarray*}
\sum_{(x,y)\in V(f)}{\frac{(p_x+p_y)}{d}} &\ge& \frac{\delta(\epsilon - d^2 \delta)}{d\cdot ImD(f)}
\end{eqnarray*}
Here $ImD$ is the image diameter defined in Definition \ref{def:imd}.
\end{lemma}

We prove the above lemma in section \ref{subsec:proofLem1}.


\subsection{Lipschitz tester}\label{sec:thmcor}
In this section we prove Theorem~\ref{thm1} and Corollary~\ref{cor:tester-for-real-range}. We first present the algorithm stated in Theorem~\ref{thm1}.
\begin{algorithm}[htb]
  \caption{Lipschitz Tester}
  \begin{algorithmic}[1]
  \REQUIRE Data domain $\T=\{0,1\}^d$, product distribution on data set $\Pi=Ber(p_1)\times Ber(p_2)\times...,\times Ber(p_d)$, failure probability parameter $\omega$, $\m{P}$-distance parameter $\epsilon^{\prime}$, discretization parameter $\delta$
  \STATE Set $\epsilon = \epsilon^{\prime} - d^2 \delta$.
  \STATE Sample $\left\lceil \frac{2}{\epsilon}\ln(\frac{2}{\omega})\right\rceil$ vertices $z_1,z_2,...,z_t$ independently from $\T$ according to the distribution $\Pi$
  \STATE Let $r=\max_{i=1}^{t}f(z_i)-\min_{i=1}^{t}f(z_i)$
  \STATE If $r>d$, reject \label{step:rej1}
  \STATE Sample $\left\lceil \frac{dr}{\delta\epsilon}\ln(\frac{2}{\omega})\right\rceil$ edges independently with each edge $(x,y)$ picked with probability $\frac{(p_x+p_y)}{d}$ from the hypercube $\T$
	\STATE If any of the sampled edges are violated, then reject, else accept \label{step:rej2}
  \end{algorithmic}
  \label{alg:tester}
\end{algorithm}

\begin{proofof}{Theorem~\ref{thm1}}
First observe that if input function $f$ is Lipschitz then the Algorithm~\ref{alg:tester} always accepts. This is because a Lipschitz function $f$ has image diameter (see Definition \ref{def:imd}) at most $d$ (and hence cannot be rejected in Step~\ref{step:rej1}. Moreover, it does not have any violated edges (and hence cannot be rejected in Step~\ref{step:rej2}). Next consider the case when $f$ is $\epsilon$-far from Lipschitz. Towards this we first extend Claim 3.1 of \cite{JR11} about sample diameter $r$ to our setting where the distance (in particular, the notion of $\epsilon$-far) is measured with respect to product distribution.
\begin{claim}
\label{clm1}
The steps 1. and 2. of the tester outputs $r\in \delta\mathbb{Z}$ such that $r\le ImD(f)$ and with probability at least $1-\frac{\omega}{2}$ (failure probability at most $\frac{\omega}{2}$), $f$ is $\epsilon$-close to having diameter that is at most $r$.
\end{claim}

\begin{proof}
Sort the points in $\set{0,1}^d$ according the function value in non-decreasing order. Let $L$ be the first $\ell$-points such that their {\em probability mass} sums up to $\frac{\epsilon}{2}$ and $R$ be the set of last $\ell'$ points such that their {\em probability mass} sums up to $\frac{\epsilon}{2}$. The rest of the proof is very similar to the proof of Claim 3.1 in \cite{JR11}, so we omit the details here.
\end{proof}

Having established Claim~\ref{clm1}, rest of the proof is identical to \cite{JR11} and we omit the details.
\end{proofof}

\begin{proofof}{Corollary~\ref{cor:tester-for-real-range}}
It is identical to the proof of Corollary 1.2 in \cite{JR11} and we omit the details.
\end{proofof}

\subsection{Repair Operator and Proof of Lemma~\ref{lem1}}
\label{subsec:MAO}
We show a transformation of an arbitrary function $f:\{0,1\}^d\rightarrow \delta\mathbb{Z}$ into Lipschitz function by changing $f$ on certain points, whose probability mass is related to the probability mass (with respect to $D_E$) of the violated edges of $\T$. This is achieved by repairing one dimension of $\T$ at a time as explained henceforth.
To achieve this, we define an \textbf{asymmetric} version of the basic operator of \cite{JR11}. The operator redefines function values so that it reduces the gap asymmetrically according to the Hamming weights (and probability masses in-turn) of the endpoints of the violated edge. This is the main difference from previous approaches (\cite{JR11}, \cite{AJMR12}) which do not work if applied directly, because of the varying probability masses of the vertices with respect to the Hamming weight. We first define the building block of the repair operator which is called the asymmetric basic operator.

\begin{definition}[Asymmetric basic operator]
\label{def:BAO1}
Given $f:\{0,1\}^d\rightarrow \delta\mathbb{Z}$, for each violated edge $\{x,y\}$ along dimension $i$, where $f(x)<f(y)-1$, define $B_i$ as follows.
\begin{enumerate}
	\item If $H(x)>H(y)$, then $B_i[f](x)=f(x)+(1-p_i)\delta$ and $B_i[f](y)=f(y)-p_i\delta$
	\item If $H(x)<H(y)$, then $B_i[f](x)=f(x)+p_i\delta$ and $B_i[f](y)=f(y)-(1-p_i)\delta$
\end{enumerate}
\end{definition}

Now we define the repair operator.

\begin{definition}[Repair operator]
\label{def:MAO1}
Given $f:\{0,1\}^d\rightarrow \delta\mathbb{Z}$, $A_i[f](x)$ is obtained from $f$ by several applications of the asymmetric basic operator (see Definition~\ref{def:BAO1}) $B_i$ along dimension $i$ followed by a single application of the rounding operator. Specifically, let $f^{\prime}$ be the function obtained from $f$ by applying $B_i$ repeatedly until there are no violated edges along the $i$-th dimension. Then, $A_i[f]$ is defined to be $\mathbf{R}[f^{\prime}]$ where the rounding operator $\mathbf{R}$ rounds the function values to the closest $\delta \mathbb{Z}$-valued function.
\end{definition}

In effect, we have the following picture for the repair operation.
\begin{align*}
f=f_0 \xrightarrow{\mathbf{R} \circ B^{\lambda_1}_1} f_{1} \xrightarrow{\mathbf{R} \circ B^{\lambda_2}_{2}} f_{2} \xrightarrow{} \cdots \xrightarrow{} f_{d-1} \xrightarrow{\mathbf{R} \circ B^{\lambda_d}_d} f_d.
\end{align*}

Now we define a measure called violation score which will be used to show the progress of repair operation. As shown later, the  violation score is approximately preserved along any dimension $j\neq i$ when we apply the repair operator to repair the edges along dimension $i$. Note that the violation score closely resembles the violation score in \cite{JR11} except that it depends on the function value as well as the probability masses of the end-points of the edge.

\begin{definition}
\label{def:VS}
The violation score of an edge with respect to function $f$, denoted by $vs(\{x,y\})$, is $\max(0,(p_x+p_y)(|f(x)-f(y)|-1))$. The violation score along dimension $i$, denoted by $VS^i(f)$, is the sum of violation scores of all edges along dimension $i$
\end{definition}

The violation score of an edge $\{x,y\}$ is positive iff it is violated and violation score of a $\delta\mathbb{Z}$ valued function is contained in the interval $\left[\delta(p_x+p_y), ImD(f)(p_x+p_y)\right]$. Let $V^i(f)$ denote be the set of edges along dimension $i$ violated by $f$. Then
\begin{eqnarray}
\label{eqn1}
\delta\cdot\sum_{\{x,y\}\in V^i(f)}(p_x+p_y) \le VS^i(f) \le \sum_{\{x,y\}\in V^i(f)}(p_x+p_y)\cdot ImD(f)
\end{eqnarray}

Lemma \ref{lem2} shows that $A_i$ does not increase the violation score in dimensions other than $i$ more than the additive value of $\delta$. The lemma makes use of the following claim.

\begin{claim}[Rounding is safe]\label{claim:basic-rounding}
Given $a, b \in \mathbb{R}$ satisfying $|a-b| \leq 1$, let $a'$ (respectively, $b'$) be the value obtained by rounding $a$ (respectively, $b$) to the closest $\delta \mathbb{Z}$ integer. Then $|a' - b'| \leq 1$. 
\end{claim}
\begin{proof}
Assume without loss of generality $ a \leq b$. For $x \in \mathbb{R}$, let $\floor{x}_{\delta}$ be the largest value in $\delta \mathbb{Z}$ not greater than $x$. Observe that $a' \in \set{\floor{a}_{\delta}, \floor{a}_{\delta}+ \delta}$. Using the fact that $\floor{a}_{\delta} \leq b' \leq \floor{a}_{\delta}+1+\delta$, we see that if $a' = \floor{a}_{\delta}+\delta$ then $|b' - a'| \leq 1$ always holds. Therefore, assume $a' = \floor{a}_{\delta}$. This can happen only if $a \leq \floor{a}_{\delta} + \delta/2$. The latter implies $b \leq \floor{a}_{\delta} + 1 + \delta/2$ (using the fact that $b -a \leq 1$). That is $b' \neq \floor{a}_{\delta}+1 + \delta$. In other words, $b' \leq  \floor{a}_{\delta}+1$ again implying $b' - a' \leq 1$, as required.
\end{proof}

\begin{lemma}
\label{lem2}
For all $i,j\in [d]$, where $i\neq j$, and every function $f:\{0,1\}^d\rightarrow \delta\mathbb{Z}$, the following holds. 
\begin{itemize}
\item \textbf{(progress)} Applying the repair operator $A_i$ does not introduce new violated edges in dimension $j$ if the dimension $j$ is violation free, i.e. $VS_j(f) = 0 \Rightarrow VS_j(A^i[f]) = 0$.
\item \textbf{(accounting)} Applying the repair operator $A_i$ does not increase the violation score in dimension $j$ by more than $\delta$, i.e. $VS_j(A^i[f])\le VS_j(f) + \delta$.
\end{itemize}
\end{lemma}

\begin{proof}
Let $f'$ be the function obtained from $f$ by applying $B_i$ repeatedly until there are no
violated edges along the $i$-th dimension. We prove the following stronger claim to prove the lemma.
\begin{claim}\label{claim:vs-is-fully-preserved}
$VS_j(f') \leq VS_j(f).$
\end{claim}
We prove the above claim momentarily but first prove the lemma using the above claim. The function $A_i[f]$ is obtained by rounding the values of $f'$ to the closest $\delta\mathbb{Z}$ values. Since rounding can never create new edge violations by Claim~\ref{claim:basic-rounding}, we immediately get the first part of the lemma. The second part follows from the observation that the rounding step modifies each function value by at most $\delta/2$. Correspondingly, the violation score of an edge along the $j$-th dimension changes by at most $2\cdot (\delta/2) \cdot (p_u + p_v)$ where the factor 2 comes because both endpoints of an edge may be rounded. Summing over all edges in the $j$-th dimension, we get, $\mbox{increase in violation score} \leq \sum_{\{u,v\}}{\delta \cdot (p_u + p_v)} = \delta$ where the last equality holds because edges along the $j$-th dimension form a perfect matching and therefore the probabilities $p_u + p_v$ sum to 1.

\begin{proofof}{Claim~\ref{claim:vs-is-fully-preserved}}
Following the proof outline of a similar proof in \cite{JR11}, we show that application of the asymmetric basic operator in dimension $i$ does not increase the violation score in dimension $j\neq i$. 
Standard arguments \cite{GoldreichGLRS00, DGLRRS99, JR11, AJMR12} show that
it is enough to analyze the effect of applying $B_i$ on one fixed disjoint square formed by adjacent edges that cross dimensions $i$ and $j$. (This is because edges along dimensions $i$ and $j$ form disjoint squares in the hypercube. So having established Claim~\ref{claim:vs-is-fully-preserved} for one fixed square of the hypercube, the full claim follows by summing up the inequalities over all such squares.)
\begin{figure}[htb!]
\includegraphics[height=20mm]{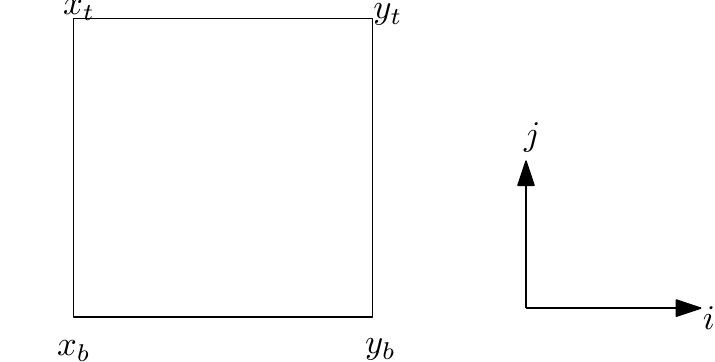}
\label{square2}
\end{figure}
Consider the two dimensional function $f:\{x_b,x_t,y_b,y_t\}\rightarrow \delta\mathbb{Z}$ where $\{x_b,x_t,y_b,y_t\}$ are positioned such that $H(y_t)=H(x_t)+1=H(y_b)+1=H(x_b)+2$ where $H(x_b)$ denotes the hamming weight of $x_b$. Assume that the basic operator is applied along the dimension $i$. We show that the violation score along dimension $j$ does not increase. Assume that the violation score along edge $\{x_b,x_t\}$ increases. First, assume that the $B_i[f](x_t)>B_i[f](x_b)$. (The other case is very similar and we will prove it later.) Then $B_i$ increases $f(x_t)$ and/or decreases $f(x_b)$. Assume that $B_i$ increases $f(x_t)$. (The other case is symmetrical.) This implies that $\{x_t,y_t\}$ is violated and $f(x_t)<f(y_t)$. Let \textbf{$f_k(x)$} (resp. \textbf{$f_k(y)$}) denote the value of $f(x)$ (resp. $f(y)$) after $k$ applications of $B_i$ on an edge $(x,y)$, for an integer $k\ge 0$. If $(x,y)$ is violated after $k-1$ applications of the basic operator, then $f_{k}(x)=f_{k-1}(x)+p_i\delta$ and $f_{k}(y)=f_{k-1}(y)-(1-p_i)\delta$ else $f_{k}(x)=f_{k-1}(x)$ and $f_{k}(y)=f_{k-1}(y)$. We will study the effect of applying $B_i$ on $(x_t, y_t)$ multiple (say $\lambda\ge 1$) times. Recall
that the repair operator is applied only if the edge is violated. This means that\

\begin{eqnarray*}
f_{\lambda-1}(x_t)&<&f_{\lambda-1}(y_t)-1\\
\Rightarrow f(x_t)+(\lambda-1) p_i \delta &<& f(y_t)-(\lambda-1) (1-p_i)\delta -1\\
\Rightarrow f(x_t)+(\lambda - 1) \delta +1 &<& f(y_t)\\
\Rightarrow f(x_t)+\lambda\delta+1 &\le& f(y_t)
\end{eqnarray*}\

The second inequality follows from the observation that since the edge is being corrected in the $\lambda^{th}$ application, it must have been corrected in all previous applications as well.   The last inequality follows from the fact that $f$ is a $\delta\mathbb{Z}$-valued function and
$\frac{1}{\delta}$ is an integer. We subtract $(1-p_i)(\lambda-1)\delta$ from both sides in the above inequality and do some rearrangement to achieve the following.

\begin{eqnarray*}
f(y_t)-(1-p_i)(\lambda-1)\delta &\ge&  f(x_t)+\lambda\delta+1-(1-p_i)(\lambda-1)\delta\\
\Rightarrow f(y_t)-(1-p_i)(\lambda-1)\delta &\ge& f(x_t)+(\lambda -1)p_i\delta +1+\delta\\
\Rightarrow f_{\lambda-1}(y_t) &\ge& f_{\lambda-1}(x_t)+1+\delta
\end{eqnarray*}

The above inequality is crucial for the remaining proof of the lemma \ref{lem1}. Now consider the cases when either the bottom edge is also violated or is not violated.

If the bottom edge is not violated then we have $f_{\lambda-1}(x_b)\ge f_{\lambda-1}(y_b)-1$ and $f_{\lambda-1}(x_b)$ and $f_{\lambda-1}(y_b)$ are not modified by the basic operator. Since $vs(\{x_t,x_b\})$ increases, $f_{\lambda-1}(x_t)>f_{\lambda-1}(x_b)+1-p_i\delta$. Combining the above inequalities, we get $f_{\lambda-1}(y_t)\ge f_{\lambda-1}(x_t)+1+\delta> f_{\lambda-1}(x_b)+2+(1-p_i)\delta\ge f_{\lambda-1}(y_b)+1+(1-p_i)\delta>f_{\lambda-1}(y_b)+1$. Thus the violation score increases along $\{x_t,x_b\}$ by $(p_{x_b}+p_{x_t}) p_i\delta$ and decreases along $\{y_b,y_t\}$ by $(p_{y_b}+p_{y_t})(1-p_i)\delta=(p_{x_b}+p_{x_t})\left(\frac{p_i}{1-p_i}\right)(1-p_i)\delta$ which is same as $(p_{x_b}+p_{x_t})p_i\delta$, keeping the violation score along the dimension $j$ unchanged.

If the bottom edge is violated, then the increase in $vs(\{x_b,x_t\})$ implies that $f_{\lambda-1}(x_b)$ must decrease (after application of $B_i$) by $p_i\delta$ (since $H(x_b)<H(y_b)$) implying $f_{\lambda-1}(y_b)+1<f_{\lambda-1}(x_b))$. Therefore $f_{\lambda-1}(x_t)+p_i\delta> f_{\lambda-1}(x_b)+1-p_i\delta$ or $f_{\lambda-1}(x_t)>f_{\lambda-1}(y_t)+1-2p_i\delta$. Therefore $f_{\lambda-1}(y_t)> f_{\lambda-1}(x_t)+1 > f(x_b)+2-2p_i\delta\ge f(y_b)+3-2p_i\delta+\delta\ge f(y_b)+1+\delta$. The last inequality is true since $\delta\le 1$ and $p_i\le 1$. Thus, $vs(\{x_t,x_b\})$ increases by at most $(p_{x_b}+p_{x_t})2p_i\delta$ while $vs(\{y_t,y_b\})$ decreases by $(p_{y_t}+p_{y_b})2(1-p_i)\delta = (p_{x_b}+p_{x_t})2p_i\delta$, ensuring that violation score along the vertical dimension does not increase.

Now we turn to the case when $B_i[f](x_t)<B_i[f](x_b)$. By the arguments very similar to the first case, it can be proved that $f_{\lambda-1}(x_t)\ge f_{\lambda-1}(y_t)+1+\delta$ and the application of basic operator decreases $f(x_t)$ by $p_i\delta$ and increases $f(y_t)$ by $(1-p_i)\delta$.

If the bottom edge is not violated then $f_{\lambda-1}(y_b)\ge f_{\lambda-1}(x_b)-1$ and $f_{\lambda-1}(x_b)$ and $f_{\lambda-1}(y_b)$ are not modified by the basic operator. Since $vs(\{x_t,x_b\})$ increases, $f_{\lambda-1}(x_b)>f_{\lambda-1}(x_t)+1- p_i\delta$. Combining the above inequalities, we get $f_{\lambda-1}(y_b)\ge f_{\lambda-1}(x_b)-1> f(x_t)-p_i\delta \ge f(y_t)+1+\delta(1-p_i)$. Thus the violation score increases along $\{x_t,x_b\}$ by $(p_{x_b}+p_{x_t}) p_i\delta$ and decreases along $\{y_b,y_t\}$ by $(p_{y_b}+p_{y_t})(1-p_i)\delta=(p_{x_b}+p_{x_t})\left(\frac{p_i}{1-p_i}\right)(1-p_i)\delta$ which is same as $(p_{x_b}+p_{x_t}) p_i\delta$, keeping the violation score along the dimension $j$ unchanged.

If the bottom edge is violated, then the increase in $vs(\{x_b,x_t\})$ implies that $f_{\lambda-1}(x_b)$ must increase implying $f_{\lambda-1}(y_b)>f_{\lambda-1}(x_b)+1$. Therefore, the increase in $vs\{x_b,x_t\}$ implies that $f_{\lambda-1}(x_b)+p_i\delta>f_{\lambda-1}(x_t)-p_i\delta+1$ or $f_{\lambda-1}(x_b)>f_{\lambda-1}(x_t)-2p_i\delta+1$. Combining the above inequalities, we get $f_{\lambda-1}(y_b)>f_{\lambda-1}(x_b)+1>f_{\lambda-1}(x_t)-2p_i\delta+2\ge f_{\lambda-1}(y_t)+3+\delta-2p_i\delta\ge f_{\lambda-1}(y_t)+1+\delta$. The last inequality is true since $\delta\le 1$ and $p_i\le 1$. Thus, $vs(\{x_t,x_b\})$ increases by at most $(p_{x_b}+p_{x_t})2p_i\delta$ while $vs(\{y_t,y_b\})$ decreases by $(p_{y_t}+p_{y_b})2(1-p_i)\delta = (p_{x_b}+p_{x_t})2p_i\delta$, ensuring that violation score along the vertical dimension does not increase.
\end{proofof}
\end{proof}

\subsubsection{Proof of Lemma \ref{lem1}}
\label{subsec:proofLem1}
Using the arguments very similar to \cite{JR11} as given below, we can get the following sequence of inequalities

\begin{eqnarray*}
Dist(f_{i-1},f_i) = Dist(f_{i-1},A_i(f_{i-1}))\le \sum_{(x,y)\in V_i(f_{i-1})}(p_x+p_y)\\
\le \frac{1}{\delta}VS^i(f_{i-1})\le \frac{1}{\delta}VS^i(f) + 2(d-i) \delta \le \frac{1}{\delta}\sum_{(x,y)\in V^i(f)}(p_x+p_y)\cdot ImD(f) + 2(d-i) \delta
\end{eqnarray*}\

Here functions $\{f_i\}_{i=0}^{i=d}$ are defined in the same way as \cite{JR11}. The first inequality holds because $A_i$ modifies $f$ only at the endpoints points $x$ and $y$ of violated edge $(x,y)$ along dimension $i$, thus paying $p_x+p_y$. The second and fourth inequalities follow from Equation (\ref{eqn1}) and the third inequality holds because of Lemma \ref{lem2}.
Therefore, by triangle inequality, we have
\begin{eqnarray*}
\label{eqn2}
Dist(f,f_d)\le \sum_{i\in [d]}Dist(f_{i-1},f_i)\le \sum_{i\in [d]}\left(\sum_{(x,y)\in V^i f(H)}(p_x+p_y)\cdot \frac{ImD(f)}{\delta} \right) + 2(d-i) \delta \\
\leq \left(\sum_{(x,y)\in V(f))}(p_x+p_y)\cdot \frac{ImD(f)}{\delta} \right) + d^2 \delta
\end{eqnarray*}\

For a function which is $\epsilon$-far from Lipschitz, we have $Dist(f,f_d)\ge \epsilon$. Therefore, from the above inequality, we have
\begin{eqnarray*}
\sum_{(x,y)\in V(f)}{\frac{(p_x+p_y)}{d}} &\ge& \frac{\delta(\epsilon - d^2 \delta)}{d\cdot ImD(f)}
\end{eqnarray*}

\section{Instantiation of privacy tester using Lipschitz testers}
\label{sec:hypergrid}
In this section, we instantiate the privacy tester of Section~\ref{sec:privacy} with both known Lipschitz testers as well as the Lipschitz tester developed in this work. This is presented in the table below. The third column gives the ``approximation factor'' as defined in Definition~\ref{def:lip-tester} for the various testers. The final column gives the privacy tester parameters that each of the tester achieves. The last row gives the result of Lipschitz tester (Section~\ref{sec:hypercube}) developed in this work.

\begin{center}
    \begin{tabular}{ | l | l | l |l | l | l|}
    \hline
    Reference 				& Functions &  Approximation factor ($\theta$) & Distribution & Tester running time & Privacy tester \\ \hline
    \cite{JR11}        	&  $\set{0,1}^d \rightarrow \mathbb{R}$      & $1+\delta$     & Uniform         &  $O(\frac{d \cdot ImD(f)}{\epsilon \delta})$                  & $(1+\delta, \alpha, \gamma, \beta)$  \\ \hline
    \cite{AJMR12}       &  $\set{1, \ldots, n}^d \rightarrow \mathbb{R}$      & $1+\delta$     &  Uniform          & $\tilde{O}\left(\frac{d \min\set{ ImD(f), nd}}{\delta\epsilon} \right)$                   &  $(1+\delta, \alpha, \gamma, \beta)$   \\ \hline
	 \cite{CS12}         &  $\set{0,1}^d \rightarrow \mathbb{R}$      & $1$     &  Uniform          &  $O(\frac{d}{\epsilon})$                  &  $(1, \alpha, \gamma, \beta)$   \\ \hline
     \textbf{This work}  &  $\set{0,1}^d \rightarrow \mathbb{R}$      & $1+\delta$    &  \textbf{Product}        &        $O\left(\frac{d \cdot ImD(f)}{(\epsilon-d^2 \delta) \delta}\right)$            &  $(1+\delta, \alpha, \gamma, \beta)$   \\ \hline
    \end{tabular}
\end{center}

%
%

\section{Discussions and Open Problems}
\label{sec:disc}

In this section we discuss about some of the interesting implications of our current work and some of the new avenues it opens up. Also we state some of the open problems that remains unresolved in our work.

\paragraph{Privacy:} In this work, we took the first step towards designing efficient testing algorithm for statistical data privacy. Our work indicates that it is indeed possible to design efficient testing algorithms for some existing notions of statistical data privacy (e.g., generalized differential privacy). It is important that the current paper should be treated as an initial study of the problem and in no way should be interpreted conclusive. It is interesting to explore other rigorous notions of data privacy, their applications and design testers for them.

In this paper, we test for generalized differential privacy, which is a relaxation of differential privacy. It remains an open problem to design a privacy tester for exact differential privacy. The problem seems to be challenging because of the fact that if we want to design an efficient tester, then usually the utility guarantees for the tester allow it to fail with some probability. Now, differential privacy being a worst case notion, it is not clear how to incorporate the failure property of the tester and yet make precise claims about differential privacy.

In the current work, we have designed privacy testers for algorithms where the domain of the data sets are either hypercube or hypergrid. A natural question that arises is that if we can extend the current results to design privacy testers when the data sets are drawn from continuous domain, unlike hypercube or hypergrid.

\paragraph{Lipschitz Testing:} This work presents the first Lipschitz property tester for the setting where the domain points are sampled from a distribution that is not uniform. Because of possible applications to statistical data privacy, this work has motivated the design of such Lipschitz testers for other domains, e.g. hypergrid. Also, this paper mainly shows the tester for the product distribution over the hypercube domain, but it still remains open to design testers for other distributions that may be correlated in some way (e.g., pairwise correlation).

\paragraph{Acknowledgements:} We would like to thank Sofya Raskhodnikova and Adam Smith for various suggestions and comments during the course of this project. 

\small
\bibliographystyle{alpha}
\bibliography{references}

\newcommand{\etalchar}[1]{$^{#1}$}
\begin{thebibliography}{AJMR12b}

\bibitem[AC06]{AilonC06}
Nir Ailon and Bernard Chazelle.
\newblock Information theory in property testing and monotonicity testing in
  higher dimension.
\newblock {\em Inf. Comput.}, 204(11):1704--1717, 2006.

\bibitem[AJMR12a]{AJMR-filters12}
Pranjal Awasthi, Madhav Jha, Marco Molinaro, and Sofya Raskhodnikova.
\newblock Limitations of local filters of lipschitz and monotone functions.
\newblock In Gupta et~al. \cite{DBLP:conf/approx/2012}, pages 387--398.

\bibitem[AJMR12b]{AJMR12}
Pranjal Awasthi, Madhav Jha, Marco Molinaro, and Sofya Raskhodnikova.
\newblock Testing lipschitz functions on hypergrid domains.
\newblock In Gupta et~al. \cite{DBLP:conf/approx/2012}, pages 387--398.

\bibitem[BBG{\etalchar{+}}11]{BBGLT11}
Raghav Bhaskar, Abhishek Bhowmick, Vipul Goyal, Srivatsan Laxman, and Abhradeep
  Thakurta.
\newblock Noiseless database privacy.
\newblock In Dong~Hoon Lee and Xiaoyun Wang, editors, {\em ASIACRYPT}, volume
  7073 of {\em Lecture Notes in Computer Science}, pages 215--232. Springer,
  2011.

\bibitem[BD12]{BD12}
Abhishek Bhowmick and Cynthia Dwork.
\newblock Natural differential privacy.
\newblock In {\em Personal communication}, 2012.

\bibitem[CKN{\etalchar{+}}11]{CKNFS11}
Joseph~A. Calandrino, Ann Kilzer, Arvind Narayanan, Edward~W. Felten, and
  Vitaly Shmatikov.
\newblock "you might also like: " privacy risks of collaborative filtering.
\newblock In {\em IEEE Symposium on Security and Privacy}, 2011.

\bibitem[CS12]{CS12}
Deeparnab Chakrabarty and C.~Seshadhri.
\newblock Optimal bounds for monotonicity and lipschitz testing over the
  hypercube.
\newblock {\em CoRR}, abs/1204.0849, 2012.

\bibitem[DGL{\etalchar{+}}99]{DGLRRS99}
Yevgeniy Dodis, Oded Goldreich, Eric Lehman, Sofya Raskhodnikova, Dana Ron, and
  Alex Samorodnitsky.
\newblock Improved testing algorithms for monotonicity.
\newblock In {\em RANDOM}, pages 97--108, 1999.

\bibitem[DKMN06]{ODO}
Cynthia Dwork, Krishnaram Kenthapadi, Frank Mcsherry, and Moni Naor.
\newblock Our data, ourselves: Privacy via distributed noise generation.
\newblock In {\em EUROCRYPT}, pages 486--503. Springer, 2006.

\bibitem[DMNS06]{DMNS06}
Cynthia Dwork, Frank McSherry, Kobbi Nissim, and Adam Smith.
\newblock Calibrating noise to sensitivity in private data analysis.
\newblock In {\em TCC}, pages 265--284, 2006.

\bibitem[Dwo06]{Dwork06}
Cynthia Dwork.
\newblock Differential privacy.
\newblock In {\em ICALP}, 2006.

\bibitem[Dwo08]{Dwork08}
Cynthia Dwork.
\newblock Differential privacy: A survey of results.
\newblock In {\em TAMC}, pages 1--19. Springer, 2008.

\bibitem[Dwo09]{Dwork09}
Cynthia Dwork.
\newblock The differential privacy frontier.
\newblock In {\em TCC}, pages 496--502. Springer, 2009.

\bibitem[GGL{\etalchar{+}}00]{GoldreichGLRS00}
Oded Goldreich, Shafi Goldwasser, Eric Lehman, Dana Ron, and Alex
  Samorodnitsky.
\newblock Testing monotonicity.
\newblock {\em Combinatorica}, 20(3):301--337, 2000.

\bibitem[GGR98a]{GGR}
Oded Goldreich, Shafi Goldwasser, and Dana Ron.
\newblock Property testing and its connection to learning and approximation.
\newblock {\em Journal of the ACM}, 45(4):653--750, 1998.

\bibitem[GGR98b]{GGR98}
Oded Goldreich, Shafi Goldwasser, and Dana Ron.
\newblock Property testing and its connection to learning and approximation.
\newblock {\em J. ACM}, 45(4):653--750, 1998.

\bibitem[GJRS12]{DBLP:conf/approx/2012}
Anupam Gupta, Klaus Jansen, Jos{\'e} D.~P. Rolim, and Rocco~A. Servedio,
  editors.
\newblock {\em Approximation, Randomization, and Combinatorial Optimization.
  Algorithms and Techniques - 15th International Workshop, APPROX 2012, and
  16th International Workshop, RANDOM 2012, Cambridge, MA, USA, August 15-17,
  2012. Proceedings}, volume 7408 of {\em Lecture Notes in Computer Science}.
  Springer, 2012.

\bibitem[GKS08]{GKS08}
Srivatsava~Ranjit Ganta, Shiva~Prasad Kasiviswanathan, and Adam Smith.
\newblock Composition attacks and auxiliary information in data privacy.
\newblock In {\em KDD}, pages 265--273, 2008.

\bibitem[GS09]{GS09}
Dana Glasner and Rocco~A. Servedio.
\newblock Distribution-free testing lower bound for basic boolean functions.
\newblock {\em Theory of Computing}, 5(1):191--216, 2009.

\bibitem[HK07]{HK07}
Shirley Halevy and Eyal Kushilevitz.
\newblock Distribution-free property-testing.
\newblock {\em SIAM J. Comput.}, 37(4):1107--1138, 2007.

\bibitem[JR11]{JR11}
Madhav Jha and Sofya Raskhodnikova.
\newblock Testing and reconstruction of lipschitz functions with applications
  to data privacy.
\newblock In Rafail Ostrovsky, editor, {\em FOCS}, pages 433--442. IEEE, 2011.

\bibitem[Kor10]{Kor2010}
Aleksandra Korolova.
\newblock Privacy violations using microtargeted ads: A case study.
\newblock In {\em ICDMW}, 2010.

\bibitem[McS09]{McSherry09}
Frank~D. McSherry.
\newblock Privacy integrated queries: an extensible platform for
  privacy-preserving data analysis.
\newblock In {\em SIGMOD}, 2009.

\bibitem[MGKV06]{MK}
Ashwin Machanavajjhala, Johannes Gehrke, Daniel Kifer, and Muthuramakrishnan
  Venkitasubramaniam.
\newblock l-diversity: Privacy beyond k-anonymity.
\newblock In {\em ICDE}, page~24, 2006.

\bibitem[MTS{\etalchar{+}}12]{MTSSC12}
Prashanth Mohan, Abhradeep Thakurta, Elaine Shi, Dawn Song, and David Culler.
\newblock Gupt: privacy preserving data analysis made easy.
\newblock In {\em SIGMOD}, 2012.

\bibitem[NRS07]{NRS}
Kobbi Nissim, Sofya Raskhodnikova, and Adam Smith.
\newblock Smooth sensitivity and sampling in private data analysis.
\newblock In {\em STOC}, 2007.

\bibitem[RP10]{ReedP10}
Jason Reed and Benjamin~C. Pierce.
\newblock Distance makes the types grow stronger: a calculus for differential
  privacy.
\newblock In {\em ICFP}, 2010.

\bibitem[RS96a]{RS}
Ronitt Rubinfeld and Madhu Sudan.
\newblock Robust characterization of polynomials with applications to program
  testing.
\newblock {\em SIAM J. Comput.}, 25(2):252--271, 1996.

\bibitem[RS96b]{RS96}
Ronitt Rubinfeld and Madhu Sudan.
\newblock Robust characterizations of polynomials with applications to program
  testing.
\newblock {\em SIAM J. Comput.}, 25(2):252--271, 1996.

\bibitem[RSK{\etalchar{+}}10]{RSKSW10}
Indrajit Roy, Srinath T.~V. Setty, Ann Kilzer, Vitaly Shmatikov, and Emmett
  Witchel.
\newblock Airavat: Security and privacy for mapreduce.
\newblock In {\em NSDI}, 2010.

\bibitem[Swe02]{Sweeney}
Latanya Sweeney.
\newblock $k$-anonymity: A model for protecting privacy.
\newblock {\em International Journal on Uncertainty, Fuzziness and
  Knowledge-based Systems}, 10(5):557--570, 2002.

\end{thebibliography}
\normalsize

\end{document}